\documentclass[10pt, a4paper]{IEEEtran}

\usepackage[switch,columnwise]{lineno}
\usepackage[english,swedish]{babel}			
\usepackage[noadjust]{cite}				
\usepackage[utf8]{inputenc}				
\usepackage{graphicx}					
\usepackage{listings} 					
\usepackage{lscape} 					
\usepackage{amsmath}					
\usepackage{amssymb}					
\usepackage{mathrsfs}					
\usepackage{amsthm}
\usepackage{url} 						
\usepackage{nameref}
\usepackage{tikz}
\usetikzlibrary{calc}
\usepackage{xcolor}
\usepackage{longtable}					
\usepackage[pdftex,bookmarks=true, colorlinks=true, linkcolor=black, urlcolor=black, citecolor=black, linktoc=section]{hyperref}
\usepackage{tikz}						
\usetikzlibrary{arrows}

\tikzstyle{int}=[draw, fill=white, minimum size=2em]
\tikzstyle{init} = [pin edge={to-,thin,black}]

\theoremstyle{plain}
\newtheorem{definition}{Definition}
\newtheorem{lemma}{Lemma}
\newtheorem{theorem}{Theorem}

\newtheorem{proposition}{Proposition}
\theoremstyle{plain}

\newtheorem{remark}{Remark}

\newtheorem{problem}{Problem}

\newcommand{\xh}{\hat{x}}
\newcommand{\xt}{\tilde{x}}

\newcommand{\yt}{\tilde{y}}
\newcommand{\zh}{\hat{z}}
\newcommand{\zt}{\tilde{z}}
\newcommand{\zb}{\bar{z}}
\newcommand{\zc}{\check{z}}
\newcommand{\sh}{\hat{s}}
\newcommand{\st}{\tilde{s}}
\newcommand{\R}{\mathbb{R}}
\newcommand{\C}{\mathbb{C}}
\newcommand{\N}{\mathbb{N}}

\newcommand{\I}{\textup{I}}
\newcommand{\T}{\intercal}
\newcommand{\tr}{\mathbf{Tr}}
\newcommand{\E}[1]{\mathbf{E}\left(#1\right)}

\newcommand{\sB}{\mathbf{B}}
\newcommand{\sS}{\mathbf{S}}
\newcommand{\sH}{\mathbf{H}}     
\newcommand{\gv}{\mathcal{N}}

\newcommand{\lp}{\left(}
\newcommand{\rp}{\right)}

\newcommand{\pr}[1]{\mathbf{Pr}\left\{#1\right\}}
\newcommand{\z}{\mathbf{z}}

\title{Feedback Capacity of Gaussian Channels Revisited} 
\date{}

\author{
	Ather Gattami
	\thanks{
		A. Gattami is with RISE AI, Research Institutes of Sweden, Box 1263, SE-164 29 Kista, Sweden. E-mail: ather.gattami@ri.se}
	\thanks{
		Copyright (c) 2018 IEEE. Personal use of this material is permitted.  However, permission to use this material for any other purposes must be obtained from the IEEE by sending a request to pubs-permissions@ieee.org.
	}
\thanks{Communicated by H. Permuter, Associate Editor for Shannon Theory.}
}

\begin{document}
\selectlanguage{english}
\maketitle

\begin{abstract}
In this paper,  we revisit the problem of finding the average capacity of the Gaussian feedback channel. 
First, we consider the problem of finding the average capacity of the analog Gaussian noise channel where
the noise has an arbitrary spectral density. We introduce a new approach to the problem where we solve
the problem over a finite number of transmissions and then consider the limit of an infinite number of transmissions. Further, we consider the important special case of stationary Gaussian noise with  finite memory. We show that the channel capacity at stationarity can be found by solving a semi-definite program, and hence computationally tractable. We also give new proofs and results of the non stationary solution which bridges the gap between results in the literature for the stationary and non stationary feedback channel capacities. It's shown that a linear communication feedback strategy is optimal. Similar to the
solution of the stationary problem, it's shown that the optimal linear strategy is to transmit a linear combination of the information symbols to be communicated and the innovations for the estimation error of the state of the noise process.
\end{abstract}

\begin{keywords}
	Gaussian channel,
	capacity,
	feedback,
	convex optimization.
\end{keywords}

\section{Introduction}

\subsection{Background and Previous Work}
We revisit the problem of communication over a Gaussian feedback channel with possibly colored noise $z=(z_1, z_2, ...)$, $z_k = 0$ for $k\le 0$
 (see Figure \ref{fig1}). More precisely, let $W\in\{1, 2, 3, ..., 2^{nR}\}$ be the message to be transmitted over the Gaussian communication channel
\[
y_k = x_k + z_k
\]
for the time horizon $k=1, ..., n$, where $x^{n} \triangleq (x_1, x_2, x_3, ..., x_{n})$ are the transmitted code words and  $y^{n} \triangleq (y_1, y_2, y_3, ..., y_{n})$, $y^0 \triangleq 0$, are the channel outputs. Each transmitted symbol $x_k$ is a deterministic function of the message $W$, the past transmitted symbols $x^{k-1}$, and the past channel outputs $y^{k-1}$ which accounts for the channel feedback. For block length $n$, we specify a $(2^{nR}, n)$ feedback code
with the encoding maps
$$
f_k: \{1, 2, ..., 2^{nR}\} \times \mathbb{R}^{k-1}
\rightarrow \mathbb{R},  ~~~ \text{for } k = 1, 2, ..., n
$$
Thus, the code words are $x_k = f_k(W, y^{k-1})$ for some time-varying function $f_k$ to be optimized. 
The transmitted symbols are subject to the average power constraint
\begin{equation}
	\label{strictP}
	\frac{1}{n}\sum_{k=1}^{n} x_k^2 \leq P.
\end{equation}

\tikzstyle{block} = [draw, fill=yellow!20, rectangle, 
minimum height=2.5em, minimum width=5em]
\tikzstyle{sum} = [draw, circle, node distance=1cm]
\tikzstyle{input} = [coordinate]
\tikzstyle{output} = [coordinate]
\tikzstyle{pinstyle} = [pin edge={to-,thin,black}]

\begin{figure}[!ht]
	\begin{tikzpicture}[auto, node distance=2cm,>=latex']
	\node [input, name=input] {};
	\node [block, right of=input, node distance=2.7cm] (code) {Code};
	\node [sum, right of=code, node distance=2.0cm] (sum) {+};
	\node [block, right of=sum] (decode) {Decode};
	\draw [draw,->] (input) -- node {$W$} (code);
	\node [output, right of=decode] (output) {};
	\node [input, above of=sum,node distance=0.8cm] (disturbance) {};
	\draw [->] (code) -- node {$x_k$} (sum);
	\draw ($(disturbance) + (0,2mm)$) node {$z_k$} ;
	\draw [->] (disturbance) -- (sum);
	\draw [->] (sum) -- node [name=y] {$y_k$} (decode);
	\draw [->] (decode) -- node [name=mhat] {$\widehat W$}(output);
	\draw [->] (y.south) -- ++(0,-1) node [below,left,yshift=-3mm] {feedback channel} -|  (code.south);
	\end{tikzpicture}
	\caption{The system studied in the paper. The feedback channel is assumed to be noiseless and the measurement noise $z_k$ is given by some colored Gaussian process with statistics known at the transmitter and receiver. }
	\label{fig1}
\end{figure}
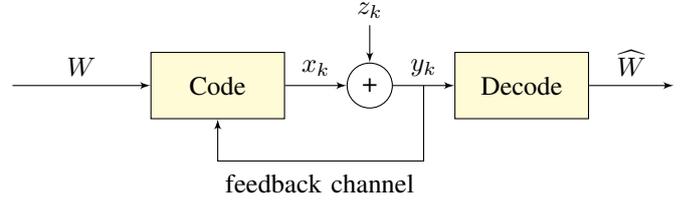

The noise process is a single-input-single-output (SISO) linear dynamical system given by 
$$z_k = \sH(\z) u_k, $$
where 
$$
\sH(\z) = \sum_{l=0}^\infty h_l\z^{-l}
$$
$h_l \in \R$, $\{u_k\}$ are independent and identically distributed (i. i. d.) Gaussian variables with $u_k\sim \gv(0, 1)$, and $\z^{-1}$ is the backward shift operator, that is $\z^{-1}u_k = u_{k-1}$. 

Let $W$ be uniformly distributed over $\{1, 2, 3, ..., 2^{nR}\}$. The decoding map $\widehat{W}_n(y^n)$ is chosen to minimize the average error probability
\begin{align*}
P^{(n)}_e 	&= \frac{1}{2^{nR}}\sum_{i=1}^{2^{nR}}  \pr{\widehat{W}_n(y^n) \ne W | W=i } \\
				&= \pr{\widehat{W}_n(y^n)\ne W}.
\end{align*}
The rate $R$ is achievable if there exists a sequence of $(2^{nR}, n)$ codes with $P^{(n)}_e \rightarrow 0$ as
$n\rightarrow \infty$. The average feedback capacity $C$ is defined as the supremum of all achievable rates.

It's  well known that for the case where the Gaussian noise is white (uncorrelated over time, that is $\sH(\z) = h_0 \in \R$), 
feedback does not improve the capacity of the channel. However, feedback can indeed increase the
capacity when the noise is colored. In the seminal work by Cover and Pombra \cite{cover:1989}, the authors introduced the quantity (where they refer to it as capacity although it does not possess the functional meaning of a
capacity)
$$
C_n = \sup_{f^n} ~ \frac{1}{n} \I(W; y^n )
$$
under the expected value of the average power constraint
\begin{equation}
	\label{expectedP}
	\frac{1}{n}\sum_{k=1}^{n} \E{x_k^2} \leq P
\end{equation}
where $\I(x ; y)$ denotes the mutual information between $x$ and $y$. It was shown that for an arbitrary Gaussian stochastic process $z$, there exists a sequence of $\left(2^{n(C_n-\epsilon)}, n\right)$ feedback codes with $P_e^{(n)}\rightarrow 0$ as $n \rightarrow \infty$ for $\epsilon > 0$. The converse holds also. It was shown that any sequence of $\left(2^{n(C_n+\epsilon)}, n\right)$ codes has $P_e^{(n)}$ bounded away from zero for all $n$. Furthermore, it was
shown that $C_n$ can be found by taking the supremum over $x^n = B_n z^n + v^n$, where $B_n\in \mathbf{T}_n$, $\mathbf{T}_n$ is the space of strictly lower triangular $n\times n$ matrices,  and $v^n$ is a
Gaussian process. That is, for $z^n \sim \gv{(0, Z_n)}$, $Z_n\succ 0$ (where $Z_n\succ 0$ means that $Z_n$ is positive definite), and $v^n  \sim \gv{(0, \mathcal{V}_n)}$, $ \mathcal{V}_n\succeq 0$ (where $\mathcal{V}_n \succeq 0$ means that $\mathcal{V}_n$ is positive semi-definite), we have
\begin{align*}
	& C_n = \\
	& \sup_{\substack{B_n\in \mathbf{T}_n\\ \mathcal{V}_n \succeq 0\\ \tr(B_nZ_nB_n^\T + \mathcal{V}_n) \le nP}} \frac{1}{2n}\log \frac{\det\left(\mathcal{V}_n + (B_n+I)Z_n(B_n+I)^\T \right)}{\det(Z_n)}
\end{align*}
The value of $C_n$ was shown to be feasible to compute using semi-definite programming \cite[Equation (2.16)]{vandenberghe1998determinant}. However, the semi-definite program size grows linearly with $n$, and it's  not possible to use the approach of \cite[Equation (2.16)]{vandenberghe1998determinant} to compute the average feedback capacity, that is $C_n$ as $n$ goes to infinity when the limit exists.
For the stricter case of average power constraints (\ref{strictP}) (and not its expected value as in (\ref{expectedP})), one needs to ensure that for $\epsilon > 0$ there exists an integer $n$ such that
\begin{equation}
\label{pcin}
\pr{\frac{1}{n} \sum_{k=1}^n x_k^2 > P} < \epsilon
\end{equation}
This can be guaranteed, for instance if the noise process $z$ is stationary (that is $R_{kl}\triangleq \E{z_k z_l}=\E{z_{j+k} z_{j+l}}=R_{(j+k)(j+l)}$ for all $j,k,l$). Note that if (\ref{pcin}) is satisfied, then as $n\rightarrow \infty$, (\ref{strictP}) and (\ref{expectedP}) become equivalent. Since we will consider stationary noise process $z$, we will consider the power constraint given by (\ref{expectedP}) (See also section VIII in \cite{cover:1989} for discussion on the stricter power inequality) 

For the case of stationary noise process $z$ with spectral density function $\sS_z(e^{i\theta})$, it was shown in  \cite[Theorem 3.2]{kim:2010}, that the average feedback capacity at stationarity, $C = \lim_{n\rightarrow \infty} C_n$, is given by
\[
C = \sup_{\sB, \sS_v} \int_{-\pi}^{\pi} \frac{1}{2}\log \frac{\sS_v(e^{i\theta}) + |\sB(e^{i\theta})+1|^2 \sS_z(e^{i\theta})}{\sS_z(e^{i\theta})} \frac{d\theta}{2\pi} 
\]
where the supremum is taken over all power spectral densities $\sS_v(\z)$ and strictly causal linear operators 
$$\sB(\z)=\sum_{l=1}^{\infty} b_l\z^{-l}$$ 
satisfying the power constraint
$$
\int_{-\pi}^{\pi} \left(\sS_v(e^{i\theta}) + |\sB(e^{i\theta})|^2 \sS_z(e^{i\theta})\right)\frac{d\theta}{2\pi}  \le P
$$

Also, \cite{kim:2010} considered the important special case of a Gaussian process $z$ of finite order, given by the state space equations
\begin{equation}
\label{linsys0}
\begin{aligned}
s_{k+1} 			&= Fs_k + Gu_k\\
z_k				&= Hs_k + u_k\\
u_k				&\sim \gv(0, 1)\\
\E{u_k u_l} 	&= 0 ~~ \forall k, l ~|~ k\neq l	\\
\end{aligned}
\end{equation}
where $F\in \R^{m\times m}$, $G\in \R^m$, $H\in \R^{1\times m}$, and $u_k$, $s_k$, and $z_k$ take values in $\R$, $\R^m$, and $\R$, respectively.
It was
shown in \cite[Theorem 6.1]{kim:2010} that the average feedback capacity at stationarity is given by
$$
C = \frac{1}{2}\log_2(Y)
$$
and $Y\in \R$ is the solution to the nonconvex optimization problem
\[
\begin{aligned}
 	\max_{\substack{X, Y \\ \Sigma \succeq 0}} &~~ Y\\
	\text{s.t. } ~~ P &\geq X\Sigma X^\T \\
				 ~~ \Sigma &= F\Sigma F^\T + GG^\T - \Gamma Y \Gamma^\T \\
				 ~~ \Gamma &= (F\Sigma (X+H)^\T + G)Y^{-1} \\
				  ~~ Y &= (X+H)\Sigma (X+H)^\T + 1
\end{aligned}
\]
The solution relies on considering the stationary problem directly instead of solving the problem over a finite horizon $n$ and then letting $n\rightarrow \infty$. The stationarity property in turn allows for using problem formulations and mathematical tools in the frequency domain.
However, a solution to the above optimization problem is intractable in practice since the equalities are nonlinear in the optimization variables. Thus, one needs a different approach in order to get a practical solution.

A related problem is that of communciation over a Gaussian channel with inter-symbol interference that was considered in \cite{yang:2005}. The inter-symbol interference was modeled as a finite order filter where the concept of directed information was used to obtain a dynamic programming formulation with constraints. Beyond the first order filter case, there is no known tractable solution to find the average feedback channel capacity numerically.

In this paper, we will use insights from systems theory in general and  linear systems theory in particular to derive new formulaes and results for the feedback capacity of Gaussian channels. Inspired by previous results for the stationary Gaussian noise channels and ideas of the innovation process approach in Kalman filter theory, we show that communicating an affine function of the innovations process of the state of the noise process is optimal.
The intuition for the structure of an optimal communication scheme is that you need to communicate a combination of the information bits and at the same time improve on the estimate of the state of the process noise. The former is important since the information bits carry the message whereas the latter reduces the uncertainty about the channel measurement noise.
Also, for the stationary noise process, we use techniques from the theories of linear systems, Riccati equations, and linear matrix inequalities to transform the optimization problem of the feedback capacity to semi-definite programming.

\subsection{Contributions}
First, we consider the problem of finding the average capacity of the analog Gaussian noise channel where
the noise has an arbitrary spectral density. We introduce a new approach where we solve
the problem over a finite number of transmissions and then consider the limit of the average capacity as the number of
transmissions tend to infinity. 
We show that the maximum average feedback capacity $C$ over an infinite time horizon can be obtained by optimizing over linear strategies $x_k=\sB(\z)z_k+v_k$  with 
$$
\sB(\z) = \sum_{l=1}^\infty b_l \z^{-l}
$$
and $v$ a  white noise process with $v_k\sim \gv(0, V)$ and $V>0$.
The capacity $C$ is the optimal value of
\begin{equation*}
	\begin{aligned}
		\sup_{\substack{V>0\\ \sB(\z)= \sum_{l=1}^\infty b_l \z^{-l}}} & \int_{-\pi}^{\pi} \frac{1}{2} \log 
		\frac{V+ |\sB(e^{i \theta})+1|^2 \sS_z(e^{i \theta}) }{
		\sS_z(e^{i \theta})
		}
		\frac{d\theta}{2\pi}\\
		\textup{s.t.  }~~~~ & V+ \int_{-\pi}^{\pi}  |\sB(e^{i \theta})|^2\sS_z(e^{i \theta})
		\frac{d\theta}{2\pi} \le P
	\end{aligned}
\end{equation*}
This is a simplification compared to the result in \cite[Theorem 3.2]{kim:2010}, where 
the capacity is optimized over white noise processes $v$ instead of all admissible 
processes with spectral densities $\sS_v(e^{i \theta})$. Also, the proofs reveal a special structure
of the optimal feedback strategy. In particular, we show that the structure of an optimal strategy is given by 
$x_k = L_k \zc_k + v_k$, where $L_k$ is some real number, $\zc_k = \zh_k - \zb_k$, $\zh_k = \E{z_k |z^{k-1}}$,  and $\zb_k=\E{\zh_k |y^{k-1}}$. Thus, an optimal communication strategy structure is shown to be linear in the stochastic variable $\zc_k$  which is the difference between
the estimate $\zh_k$ of the noise $z_k$ based on the previous measurements of the noise process 
$z^{k-1}$ and the estimate of $\zh_k$ based on the channel output measurements $y^{k-1}$. This strategy is identical to the one found in \cite{kim:2010} for the \textit{stationary} case, with the simplification that we can restrict $v$ to be a white Gaussian process.

We then consider the important special case of stationary Gaussian noise with finite memory. We show that the average feedback channel capacity at stationarity can be found by solving a semi-definite program, and hence computationally tractable. 
In particular, we show that the channel capacity is given by
$$
C = \frac{1}{2} \log_2(Y)
$$
where $Y\in \R$ is the optimal solution to the convex optimization problem
\begin{equation*}
\begin{aligned}
\sup_{\substack{K \\ \Sigma\succ 0}} &~~   Y \\
	\textup{s.t. }& \\   0 \prec & \begin{pmatrix}
										P		&  K\\
										K^\T	& \Sigma
									\end{pmatrix}\\
					~~	 0 \preceq & 
					\begin{pmatrix}
						F\Sigma F^\T - \Sigma + GG^\T 		& FK^\T + F\Sigma H^\T + G\\
						(FK^\T + F\Sigma H^\T + G)^\T	& Y
					\end{pmatrix} \\
					Y=&  K H^\T + HK^\T + H\Sigma H^\T + P + 1
\end{aligned}
\end{equation*}

\subsection{Paper Outline}
In section \ref{prel}, we introduce the notation used in this paper and give some known results from 
system theory and information theory. In section \ref{finitesect}, we formulate the general problem of
finding the average capacity of the Gaussian channel with feedback where we derive results similar to those
obtained in \cite{cover:1989} and \cite{kim:2010} for the non stationary and stationary Gaussian noise processes,
respectively. We then consider the special case of the stationary Gaussian noise process of finite order in sections \ref{finiteorderfinite} and \ref{finiteorderstationary}.
There, we give new proofs that also show that we can find the average feedback capacity for this case by solving a semi-definite program, and hence
making it computationally tractable. Finally, we provide an example in section \ref{example} with a Matlab code in the appendix that can be used to compare 
the results of this paper with exisiting solutions for the first order Gaussian process case.
Most of the proofs are relegated to the appendix.

\section{Preliminaries}
\label{prel}

\subsection{Notation}

\begin{tabular}{ll}
$\mathbb{N}$ & The set of positive integers.\\
$\mathbb{R}$ & The set of real numbers.\\
$\mathbb{C}$ & The set of complex numbers.\\
$\mathbb{S}^n$ & The set of $n\times n$ symmetric matrices.\\
$\mathbb{S}^n_{+}$& The set of $n\times n$ symmetric positive\\
& semidefinite matrices.\\
$\mathbb{S}^n_{++}$	& The set of $n\times n$ symmetric positive\\
					& definite matrices.\\
$\succeq$ & $A\succeq B$ $\Longleftrightarrow$ $A-B\in \mathbb{S}^n_{+}$.\\
$\succ$ & $A\succ B$ $\Longleftrightarrow$ $A-B\in \mathbb{S}^n_{++}$.\\
$A^{\dagger}$ 	& The Moore-Penrose pseudo-inverse\\ & of the matrix $A$.\\
$s^k$ 	& $s^k=(s_1, s_2, ..., s_k)$ and $s^0\triangleq 0$.\\
$|s|$ 	& For a vector $s=(s_1, s_2, ..., s_k)$, \\
		& $|s|^2 = \sum_{i=1}^k |s_i|^2$.
\end{tabular}

\subsection{System Theory}
The material here can be found in \cite{zhou:1996}.

\begin{definition}[Detectability]
Let $F\in \R^{m\times m}$ and $H\in \R^{p\times m}$.  The pair of matrices $(H, F)$ is detectable 
if 
$$\begin{pmatrix} F-\lambda I \\ H \end{pmatrix}$$ 
has full column rank for all $\lambda\in \C$ such that $|\lambda| \geq 1$. 
\end{definition}


\begin{definition}[Controllability]
Let $F\in \R^{m\times m}$ and $G\in \R^{m\times r}$.  The pair of matrices $(F, G)$ is controllable 
if 
$$\begin{pmatrix} F-\lambda I & G \end{pmatrix}$$ 
has full row rank for all $\lambda\in \C$.
\end{definition}

\begin{definition}[Stability]
Let $F\in \R^{m\times m}$. The matrix $F$ is stable if and only if  
its eigenvalues have modulus strictly less than 1.
\end{definition}


\subsection{Optimal Estimation of Gaussian Processes}
Consider a Gaussian process $z$ given by the state space equations in (\ref{linsys0}). 
Let $\sh_k = \mu_k(z^{k-1})$ be an estimate of $s_k$ based on the measurements $z^{k-1}$ and  let $\st_k = s_k - \sh_k$.
Suppose that we want to minimize the average squared estimation error
$$
\frac{1}{n} \sum_{k=1}^{n} \E{|\st_k|^2}. 
$$
It's  well known (\cite{kailath:linsys}) that the optimal estimator is given by $\sh_k = \E{s_k | z^{k-1}}$ which obeys the Kalman filter recursions
\begin{equation}
\label{kalmanfilter}
	\begin{aligned}
		S_{k} 		&= \E{\st_k\st_k^\T}\\
		S_1		&= 0\\
		S_{k+1}	&= FS_k F^\T + GG^\T \\ 
					&~~~ -  (FS_k H^\T + GE^\T)(HS_k H^\T + 1)^{-1} \\
					& ~~~~~~\times (HS_kF^\T + G^\T)\\
		K_{k}		&= (FS_k H^\T + G)(HS_k H^\T +1)^{-1}\\
		\sh_{k+1}&= F\sh_k +K_k(z_k - H\sh_k)\\
		\st_{k+1} &= (F-K_kH)\st_k + Gu_k - K_k u_k	\\	
		\zt_k 		&= H\st_k + u_k
	\end{aligned}
\end{equation}
A property of the Kalman filter is that the innovations $\zt_k = z_k - H\sh_k = H\st_k + u_k$ are independent for all $k\in \N$.


\subsection{Entropy Properties of Gaussian Variables and Processes}


The differential entropy of the Gaussian process given by (\ref{linsys0}) over a time horizon $k=1, ..., n$ is $h(z^n)$ which may be rewritten as a sum of conditional differential entropies using the differential entropy chain rule  
$$
h(z^n) = \sum_{k=1}^n h(z_k|z^{k-1})
$$
with $z^{k} \triangleq 0$ if $k=0$. The differential entropy rate is given by
$$
h(z) = \lim_{n\rightarrow \infty} \frac{1}{n} h(z^n).
$$

\begin{proposition}
\label{entropy}
	Consider a Gaussian process $z$ with spectral density function $\sS_z(\z)$. Then, the differential entropy rate of $z$ is given by
	$$
	h(z) = \int_{-\pi}^{\pi} \frac{1}{2}\log{\left(\sS_z \left(e^{i \theta}\right)\right)}\frac{d\theta}{2\pi}.
	$$
\end{proposition}
\begin{proof}
Consult \cite{cover:2006}.
\end{proof}
\begin{proposition}
\label{ssRicc}
Consider a Gaussian process given by (\ref{linsys0}) over a finite time horizon $k=1, ..., n$.
The differential entropy of $z^n$ is given by
$$
h(z^n) = \frac{1}{2}\sum_{k=1}^n \log_2 \left( 2\pi e(HS_kH^\T + 1)\right) 
$$
where $S_k$ is given by the recursion
\begin{equation}
\label{riccRecursion}
\begin{aligned}
S_1		&= \E{s_1s_1^\T}\\
		S_{k+1}	&= FS_k F^\T + GG^\T \\ 
					&~~~ -  (FS_k H^\T + G)(HS_k H^\T +1)^{-1} \\
					& ~~~~~~\times (HS_kF^\T + G^\T)\\
\end{aligned}
\end{equation}

Furthermore, if $(H, F)$ is detectable, then the stationary differential entropy rate is given by
$$
h(z) = \lim_{n\rightarrow \infty} \frac{1}{n} h(z^n) = \frac{1}{2}\log_2 \left( 2\pi e(HS H^\T + 1)\right) 
$$
where $S$ is the unique solution to the Riccati equation
\begin{equation}
\label{ssRiccati}
\begin{aligned}
		S	&= FSF^\T + GG^\T \\ 
					&~~~ -  (FS H^\T + G)(HS H^\T + 1)^{-1} \\
					& ~~~~~~\times (HSF^\T + G^\T)\\
\end{aligned}
\end{equation}

\end{proposition} 
\begin{proof}
See Appendix \ref{prp2}.
\end{proof}



\section{Feedback Capacity of Gaussian Channels with Colored Noise of Arbitrary Order}
\label{finitesect}

In this section, we will study the general problem of finding the feedback capacity with respect to Gaussian noise with an arbitrary spectral density function. More precisely, we consider the following problem.

\begin{problem}
\label{problem1}
Let $W$ be the message to be transmitted and consider the Gaussian communication channel
\[
y_k = x_k + z_k
\]
over a time horizon $n$, where $x_k = f_k(W, y^{k-1})$ is the transmitted signal over the channel 
with the average power constraint
\[
\frac{1}{n}\sum_{k=1}^{n} \E{x_k^2} \leq P,
\]
where $y$ is the measurement signal at the receiver, and $z$ is the stationary Gaussian measurement noise process with spectral density function $\sS_z(\z)$. 
Find the value of
$$
C_n = \sup_{f^n} ~ \frac{1}{n} \I(W; y^n ).
$$
where $f^n$ are deterministic functions. In particular, find the average feedback channel capacity
$C = \lim_{n \rightarrow \infty} C_n$.
\end{problem}

It has been shown in \cite{cover:1989} that for a noise sequence $z^n$, 
the mutual information $ \I(W; y^n)$ is given by
$$
\I(W; y^n) = h(y^n) - h(z^n)
$$
To make this paper self contained, we state this result.
\begin{proposition}
\label{coverRes}
Consider the Gaussian feedback channel as described in Problem \ref{problem1}. Then,
$$
\I(W; y^n) = h(y^n) - h(z^n)
$$
\end{proposition}
\begin{proof}
See Appendix \ref{prp3}.
\end{proof}
For a Gaussian noise sequence $z^n$ with zero mean and covariance $Z_n$, it was further shown in \cite{cover:1989} that the optimal input distribution of the sequence $x^n$ that maximizes $C_n$ is obtained for $x^n = B_n z^n + v^n$,   
where $B_n\in \R^{n\times n}$ is strictly lower triangular,  $v^n$ is a Gaussian sequence with covariance $\mathcal{V}_n\in \mathbb{S}^n_{+}$, and the pair $(B_n, \mathcal{V}_n)$ satisfies the power constraint
$$
\tr(B_n Z_nB_n^\T + \mathcal{V}_n) = \E{|x^n|^2} \leq nP
$$

Proposition \ref{coverRes} gives the relation $\I(W ; y^n)  =  h(y^n) - h(z^n)$.
Thus, the mutual information between $W$ and $y^n$ satisfies
\[
\begin{aligned}
	 	\I(W ; y^n) 
		&=  h(y^n) - h(z^n) \\
		&\le  \frac{1}{2} \log_2 \frac{\det(\mathcal{V}_n + (B_n+I)Z_n(B_n+I)^\T)}{\det (Z_n)} \\
		&=  \frac{1}{2} \log_2 \det(\mathcal{V}_n + (B_n+I)Z_n(B_n+I)^\T) \\
		&~~~ - \frac{1}{2} \log_2 \det(Z_n)\\
		&\triangleq \mathcal{I}(B_n, \mathcal{V}_n, Z_n)
\end{aligned}
\]

The value of $C_n$ over the time horizon $n$ is given by 
\[
C_n  = \sup_{\substack{B_n\\ \mathcal{V}_n\succeq 0\\ \tr(B_n Z_nB_n^\T + \mathcal{V}_n) \leq nP}} \frac{1}{n}  \mathcal{I}(B_n, \mathcal{V}_n, Z_n)
\]
where the maximum is taken over $B_n\in \R^{n\times n}$ being strictly lower triangular. 

\begin{theorem}
\label{infhorizon_thm}
The average feedback capacity $C$ in Problem \ref{problem1} over an infinite time horizon can be obtained by optimizing over linear strategies $x_k=\sB(\z)z_k+v_k$  with 
$$
\sB(\z) = \sum_{l=1}^\infty b_l \z^{-l},
$$
$v_k\sim \gv(0, V)$, and $V>0$. The capacity $C$ is 
the optimal value of
\begin{equation}
\label{infhorizon_capacity}
	\begin{aligned}
		\sup_{\substack{V>0\\ \sB(\z)= \sum_{l=1}^\infty b_l \z^{-l}}} & \int_{-\pi}^{\pi} \frac{1}{2} \log 
		\frac{V+ |\sB(e^{i \theta})+1|^2 \sS_z(e^{i \theta}) }{
		\sS_z(e^{i \theta})
		}
		\frac{d\theta}{2\pi}\\
		\textup{s.t.  }~~~~ & V+ \int_{-\pi}^{\pi}  |\sB(e^{i \theta})|^2\sS_z(e^{i \theta})
		\frac{d\theta}{2\pi} \le P
	\end{aligned}
\end{equation}
In particular, the structure of an optimal strategy is given by $x_k = L_k \zc_k + v_k$, where
$\zc_k = \zh_k - \zb_k$, $\zh_k = \E{z_k |z^{k-1}}$,  and $\zb_k=\E{\zh_k |y^{k-1}}$.

\end{theorem}
\begin{proof}
See Appendix \ref{thm1}.
\end{proof}
Theorem \ref{infhorizon_thm} provides a somewhat simplified formulation of the Gaussian average feedback channel capacity  compared to the result in \cite[Theorem 3.2]{kim:2010}, where we find that  
the capacity is optimized over white noise processes $v$ instead of all admissible 
processes with spectral densities $\sS_v(e^{i \theta})$ as in \cite{kim:2010}. Also, an optimal communication strategy structure is revealed, and it's linear in the stochastic variable $\zc_k$ which is the difference between the 
estimate $\zh_k$ of the noise $z_k$, based on the previous measurements of the noise process 
$z^{k-1}$, and the estimate of $\zh_k$ based on the channel output measurements $y^{k-1}$.

\section{Gaussian Channel with Noise Process of Finite Order}
\label{finiteorderfinite}
Let $z$ be a Gaussian process of finite order given by the state space equations (\ref{linsys0}). 
In \cite{kim:2010}, by using the results of \cite{cover:1989} above, it was shown that the optimal affine strategy of the transmitted symbols $x_k$ to maximize the capacity at stationarity (that is, as
$n\rightarrow \infty$), is given by
$$
x_k = X(s_k - \E{s_k |y^{k-1}}) +v_k
$$
where $v$ is a white Gaussian process independent of $u$ with $v_k\sim \gv(0,V)$, and $X^\T$ is some vector in $\R^{m}$. In fact, it was shown that $v=0$ is optimal, but we will keep it as an optimization parameter as it will simplify the optimization problem considerably. 

We give here and in section \ref{finiteorderstationary} new (and shorter) proofs that summarize the results of \cite{cover:1989} and \cite{kim:2010} on the value of $C_n$ and the value of the average feedback capacity $C$.

\begin{theorem}
\label{optimalTx}
The average feedback capacity in Problem \ref{problem1} over a finite horizon $k=1, ..., n$ with the noise process $z$ given by (\ref{linsys0})
is achieved by $x_k = X_k\tilde{s}_k + v_k$ for some set of vectors $\{X_k^\T\}$, where $\sh_k=\E{s_k |y^{k-1}}$, $\st_k = s_k - \sh_k$, and
$\{v_k\}$ is a white Gaussian process with $v_k\sim \gv(0, V_k)$, independent of $u$. The value of $C_n$ is given by
\begin{equation}
\label{sumCap}
C_n = \sup_{\substack{X_1, ..., X_n\\ V_1, ..., V_n\ge 0}} \frac{1}{2n}\sum_{k=1}^n \log_2(Y_k) 
\end{equation}
subject to $\Sigma_1 = 0$,
$$P\ge X_k\Sigma_k X_k^\T + V_k, $$
\begin{equation}
\label{Yrec}
Y_k = (X_k+H)\Sigma_k (X_k+H)^\T +  V_k + 1,
\end{equation}
\begin{align}
\label{Gammarec}
\Gamma_k 
		&= (F\Sigma_k (X_k +H)^\T + G)Y_k^{-1}
\end{align}
and
\begin{equation}
\label{riccrec}
\begin{aligned}
\Sigma_{k+1} 	&= (F-\Gamma_k(X_k+H))\Sigma_k (F-\Gamma_k(X_k+H))^\T \\
			& ~~~ + (G-\Gamma_k)(G-\Gamma_k)^\T + \Gamma_k V_k \Gamma_k^\T\\ 
			&= F\Sigma_k F^\T + GG^\T - \Gamma_k Y_k \Gamma_k^\T 
\end{aligned}
\end{equation}

\end{theorem}
\begin{proof}
See Appendix \ref{thm2}.
\end{proof}

\begin{remark}
\textup{The recursive optimization problem in Theorem \ref{optimalTx} given by equations (\ref{sumCap})--(\ref{riccrec}) can be solved 
by introducing a Lagrange multiplier corresponding to the power constraint and use the bisection method with respect to that 
Lagrange mulitplier and then solve the optimization problem using dynamic programming for each fixed value of the Lagrange mulitplier. For further details, consult \cite{gattami:tac:10}.}
\end{remark}

\section{Feedback Capacity with Process Noise of Finite Order at Stationarity}
\label{finiteorderstationary}
Now consider a stationary finite order Gaussian noise process $z$ given by
\begin{equation}
\label{linsys}
\begin{aligned}
s_{k+1} 			&= Fs_k + Gu_k\\
z_k					&= Hs_k + u_k\\
u_k				&\sim \gv(0, 1)\\
\E{u_k u_l} 	&= 0 ~~ \forall k, l ~|~ k\neq l	\\
\end{aligned}
\end{equation}
where $u_k$, $s_k$, and $z_k$ take values in $\R$, $\R^m$, and $\R$, respectively, and we have replaced $s_1 = 0$ in (\ref{linsys0})
with its stationary distribution. 
From linear system theory \cite{zhou:1996}, we know that the pair $(H, F)$ must be detectable in order for the variance of the estimation error $z_k - \E{z_k\mid z^{k-1}}$ to be bounded (which is necessary for the capacity to be positive otherwise the signal to noise ratio will be zero). Since we assumed that $F$ is stable, the pair $(H, F)$ will be detectable. 

The following result states how to compute the Gaussian average feedback channel capacity by using finite-dimensional nonlinear programming. This result is similar to the one given in \cite{kim:2010}, and the computational complexity of the nonlinear programming (nonlinear semidefinite programming) problem is NP-hard in general. 
\begin{theorem}
\label{kim}
The average feedback channel capacity in Problem \ref{problem1}, with
the process $z$ given by (\ref{linsys}), is 
$$C = \frac{1}{2}\log_2(Y)$$
where $Y$ is the solution to the optimization problem
\begin{equation}
\label{nonnegV}
\begin{aligned}
 	\sup_{V>0, Y} \sup_{\substack{X \\ \Sigma \succeq 0}} &~~ Y\\
	\textup{s.t. }   ~~ P &\ge X\Sigma X^\T + V \\
				 ~~ \Sigma &= F\Sigma F^\T + GG^\T - \Gamma Y \Gamma^\T \\
				 ~~ \Gamma &= (F\Sigma (X+H)^\T + G)Y^{-1} \\
				 ~~ Y &= (X+H)\Sigma (X+H)^\T +  V + 1
\end{aligned}
\end{equation}
\end{theorem}
\begin{proof}
See Appendix \ref{thm3}.
\end{proof}

Note that taking $V=0$ renders the same (nonconvex) optimization problem as that in \cite{kim:2010}. However, this case
is not achievable since the mutual information vanishes for $V=0$.

\begin{lemma}
\label{epslemma}
Optimization problem (\ref{nonnegV}) is equivalent to 
\begin{equation}
\label{ineqricc}
\begin{aligned}
 	\sup_{V>0} \sup_{\substack{X \\ \Sigma \succeq 0}} &~~ Y \\
	\text{s.t. }   ~~ P &\ge X\Sigma X^\T + V \\
				 ~~ \Sigma &\preceq  F\Sigma F^\T + GG^\T - \Gamma Y \Gamma^\T \\
				 ~~ \Gamma &= (F\Sigma (X+H)^\T + G)Y^{-1} \\
				 ~~ Y &= (X+H)\Sigma (X+H)^\T +  V + 1
\end{aligned}
\end{equation}
\end{lemma}
\begin{proof}
See Appendix \ref{lem1}.
\end{proof}
Now we turn to the recursion 
\begin{equation*}
\begin{aligned}
\Sigma_{k+1} 	&= (F-\Gamma_k(X_k+H))\Sigma_k (F-\Gamma_k(X_k+H))^\T \\[2mm]
& ~~~ + (G-\Gamma_k)(G-\Gamma_k)^\T + \Gamma_k V_k \Gamma_k^\T\\[2mm]
&= F\Sigma_k F^\T + GG^\T - \Gamma_k Y_k \Gamma_k^\T \\[2mm]
\Gamma_k 
&= (F\Sigma_k (X_k +H)^\T + G)Y_k^{-1}
\end{aligned}
\end{equation*}
at stationarity where for all $k$,
$$Y_k = Y$$
$$X_k=X,$$
$$V_k=V,$$
$$\Gamma_k = \Gamma,$$
$$\Sigma_k = \Sigma,$$
and $\Sigma$ is the unique solution to the Riccati equation
\begin{equation*}
\begin{aligned}
	\Sigma 	&= (F-\Gamma(X+H))\Sigma (F-\Gamma(X+H))^\T \\
	& ~~~ + (G-\Gamma)(G-\Gamma)^\T + \Gamma V \Gamma^\T\\ 
	&= F\Sigma F^\T + GG^\T - \Gamma Y \Gamma^\T 
\end{aligned}
\end{equation*}
with
$$
\Gamma = (F\Sigma (X+H)^\T + G)Y^{-1}.
$$
We will utilize that $V>0$ to show that the pair
$$
(F-\Gamma(X+H), (G-\Gamma)(G-\Gamma)^\T + \Gamma V \Gamma^\T)
$$
is controllable. 
\begin{lemma}
\label{controllablepair}
Suppose that $(F, G)$ is controllable and $V>0$. Then, the pair
$$
(F-\Gamma(X+H), (G-\Gamma)(G-\Gamma)^\T + \Gamma V \Gamma^\T)
$$
is controllable.
\end{lemma}
\begin{proof}
See Appendix \ref{lem2}.
\end{proof}
\begin{lemma}
\label{stableF}
Suppose that $(F, G)$ is controllable, $V>0$, and
\begin{equation*}
\begin{aligned}
\Sigma 	&= (F-\Gamma(X+H))\Sigma (F-\Gamma(X+H))^\T \\
			& ~~~ + (G-\Gamma)(G-\Gamma)^\T + \Gamma V \Gamma^\T\\ 
\end{aligned}
\end{equation*}
Then, $F-\Gamma(X+H)$ is stable.
\end{lemma}
\begin{proof}
See Appendix \ref{lem3}.
\end{proof}

\begin{lemma}
\label{invertibleS}
Suppose that $(F, G)$ is controllable and $V>0$ and let $\Sigma$ satisfy the Riccati equation
\begin{equation*}
\begin{aligned}
\Sigma 	&= (F-\Gamma(X+H))\Sigma (F-\Gamma(X+H))^\T \\
			& ~~~ + (G-\Gamma)(G-\Gamma)^\T + \Gamma V \Gamma^\T\\ 
\end{aligned}
\end{equation*}
Then, $\Sigma \succ 0$.
\end{lemma}
\begin{proof}
See Appendix \ref{lem4}.
\end{proof}

Since Lemma \ref{invertibleS} established the invertibility of $\Sigma$, we see that optimization problem (\ref{nonnegV}) is equivalent to that of optimizing over strictly positive definite matrices $\Sigma$, that is
\begin{equation}
\label{posSigma}
\begin{aligned}
 	\sup_{V>0} \sup_{\substack{X \\ \Sigma \succ 0}} &~~ Y\\
	\text{s.t. }	~~ P &\ge X\Sigma X^\T + V \\
				~~ \Sigma &\preceq F\Sigma F^\T + GG^\T - \Gamma Y \Gamma^\T \\
				~~ \Gamma &= (F\Sigma (X+H)^\T + G)Y^{-1} \\
				~~ Y &= (X+H)\Sigma (X+H)^\T +  V + 1
\end{aligned}
\end{equation}

Optimization problem (\ref{posSigma}) can now be transformed to a semi-definite program.
The trick is to first eliminate the dependence on $V$ and obtain the desired inequalities instead of equalities. Then, by making a variable substitution according to $K = X\Sigma$, which is possible since $\Sigma$ is invertible, we will be able to use a Schur complement argument in order to transform the constraints in (\ref{posSigma}) into a set of linear matrix inequalities (LMI:s). 

We are now ready to state the main result of this paper. 
\begin{theorem}
\label{main}
The average feedback capacity of the Gaussian channel is given by 
$$C = \frac{1}{2}\log_2 \left(Y \right)$$ 
where $Y$ is the optimal solution of 
\begin{equation}
\label{maxminopt}
\begin{aligned}
 	\sup_{\substack{K \\ \Sigma\succ 0}} &~~   Y \\
	\textup{s.t. }& \\   0 \prec & \begin{pmatrix}
										P		&  K\\
										K^\T	& \Sigma
									\end{pmatrix}\\
					~~	 0 \preceq & 
					\begin{pmatrix}
						F\Sigma F^\T - \Sigma + GG^\T 		& FK^\T + F\Sigma H^\T + G\\
						(FK^\T + F\Sigma H^\T + G)^\T	& Y
					\end{pmatrix} \\
					Y=&  K H^\T + HK^\T + H\Sigma H^\T + P + 1
\end{aligned}
\end{equation}
\end{theorem}
\begin{proof}
See Appendix \ref{thm4}.
\end{proof}

Theorem \ref{main} shows that the Gaussian feedback channel capacity can be computed
efficiently using semidefinite programming as opposed to the nonlinear programming formulation for the computation of the feedback channel capacity in \cite{kim:2010}, which is not tractable in practice based on the optimization tools to date. 

\section{Example}
\label{example}
Consider a communication feedback channel with Gaussian noise given by a first order process according to
$$z_k + \beta z_{k-1} = u_k + \alpha u_{k-1}$$
with state space representation
\begin{align*}
s_{k+1} 	&= -\beta s_k + u_k\\
z_k 		&= (\alpha - \beta)s_k + u_k\\
y_k 		&= x_k + z_k
\end{align*}
$\alpha\in [-1, 1]$, $\beta \in (-1, 1)$, $u_k\sim \gv(0,1)$, and
\[ \sigma = \textup{sign}(\beta -\alpha) =
  \begin{cases}
    1       & \quad \textup{if } \beta > \alpha\\
    0       & \quad \textup{if } \beta = \alpha\\
    -1  & \quad \textup{if } \beta < \alpha
  \end{cases}
\]

In \cite{kim:2010}, it was shown that the feedback capacity of the above channel with a power constraint $\E{x_k^2}\le P$
is given by
$$
-\log_2(r)
$$
where $r$ is the unique positive real root of the polynomial equation
$$
(\alpha^2+\beta^2 P)r^4 + 2\sigma(\alpha+\beta P)r^3+(P+1-\alpha^2)r^2 - 2\sigma \alpha r - 1 = 0
$$
This can be compared to the solution of the semi-definite optimization problem (\ref{maxminopt}). 
By noting that $F = -\beta$, $G = 1$, and $H = \alpha - \beta$, one can verify numerically that the channel capacity 
$\frac{1}{2}\log_2(Y)$ coincides with that of the polynomial solution $-\log_2(r)$(see the Matlab code in Appendix \ref{code}, where we used CVX, a package for specifying and solving convex programs \cite{cvx, gb08}).

\section{Conclusion}
We considered the problem of finding the average capacity of the Gaussian noise channel where
the noise has an arbitrary spectral density. We introduced a new approach to the problem where we solved
the problem over a finite number of transmissions and then considered the limit of the average capacity as the number of
transmissions went to infinity. 
We also provided new proofs and results of the non stationary solution which bridges the gap between results in the literature \cite{cover:1989, kim:2010}. For the stationary Gaussian noise, 
we showed that the average feedback capacity $C$ over an infinite time horizon can be obtained by optimizing over linear strategies $x_k=\sB(\z)z_k+v_k$  with 
$$
\sB(\z) = \sum_{l=1}^\infty b_l \z^{-l}
$$
$v_k\sim \gv(0, V)$, and $V>0$. The capacity $C$ is 
the optimal value of
\begin{equation*}
	\begin{aligned}
		\sup_{\substack{V>0\\ \sB(\z)= \sum_{l=1}^\infty b_l \z^{-l}}} & \int_{-\pi}^{\pi} \frac{1}{2} \log 
		\frac{V+ |\sB(e^{i \theta})+1|^2 \sS_z(e^{i \theta}) }{
		\sS_z(e^{i \theta})
		}
		\frac{d\theta}{2\pi}\\
		\textup{s.t.  }~~~~ & V+ \int_{-\pi}^{\pi}  |\sB(e^{i \theta})|^2\sS_z(e^{i \theta})
		\frac{d\theta}{2\pi} \le P
	\end{aligned}
\end{equation*}
In particular, we showed that the structure of an optimal strategy is given by $x_k = L_k \zc_k + v_k$, where
$\zc_k = \zh_k - \zb_k$, $\zh_k = \E{z_k |z^{k-1}}$,  and $\zb_k=\E{\zh_k |y^{k-1}}$.
Thus, an optimal communication strategy structure was revealed. An optimal strategy is linear in the stochastic variable $\zc_k$ which is the difference between
the the estimate $\zh_k$ of the noise $z_k$ based on the previous measurements of the noise process 
$z^{k-1}$ and the estimate of $\zh_k$ based on the channel output measurements $y^{k-1}$.

We also considered the special case of stationary Gaussian noise with finite memory. 
We showed that the channel capacity at stationarity can be found by solving a semi-definite program, and hence computationally tractable. 
In particular, we showed that the average feedback channel capacity is given by
$$
C = \frac{1}{2} \log_2(Y)
$$
where $Y\in \R$ is the optimal solution to the optimization problem
given by (\ref{maxminopt}).

\section*{Acknowledgment}
The author is grateful to the Associate Editor and the reviewers for constructive feedback and for useful comments and suggestions.


\appendix

\subsection{Results from System Theory}

\begin{proposition}
\label{riccinv}
Let $Q\succeq 0$ and $F$ stable. Then, the unique positive semi-definite solution $\Sigma$ to the Lyapunov equation
$$
\Sigma = F\Sigma F^\T + Q
$$
is invertible if $(F, Q)$ is controllable.
\end{proposition}
\begin{proof}
Consult \cite{zhou:1996}.
\end{proof}

\begin{proposition}
\label{contcond}
Let $F\in \R^{m\times m}$ and $G \in \R^{m\times r}$. Then, the pair $(F, G)$ is controllable if and only if there does not exist a vector $x\ne 0$ and a scalar $\lambda\in \C$ such that $x^\T F = \lambda x^\T $ and $x^\T G = 0$.
\end{proposition}
\begin{proof}
Consult \cite{zhou:1996}.
\end{proof}

\subsection{Proof of Proposition \ref{ssRicc}}
\label{prp2}
Note that  $\sh_k = \E{s_k | z^{k-1}}$ and $\st_k = s_k - \sh_k$ are given by the Kalman filter (\ref{kalmanfilter}), so $z_k - \E{z_k | z^{k-1}} = \zt_k $
is the Gaussian output estimation error given by the Kalman filter (\ref{kalmanfilter}).
Thus, the differential entropy chain rule gives the equality
$$
h(z^n) = \sum_{k=1}^n h(\zt_k) = \frac{1}{2}\sum_{k=1}^n \log_2 \left( 2\pi e(HS_kH^\T + 1)\right) 
$$
Since the pair $(H, F)$ is detectable, the recursion (\ref{riccRecursion}) has a unique stationary positive semi-definite solution $S_k = S$. The stationary differential entropy rate is then given by
$$
h(z) = \lim_{n\rightarrow \infty} \frac{1}{n} h(z^n) = \frac{1}{2}\log_2 \left( 2\pi e(HS H^\T + 1)\right) 
$$
where $S$ is the unique positive semi-definite solution to the Riccati equation (\ref{ssRiccati}). \qed

\subsection{Proof of Proposition \ref{coverRes}}
\label{prp3}
\begin{align}
\I(W; y^n) 	&= h(y^n) - h(y^n\mid W) \nonumber\\
			&= h(y^n) - \sum_{k=1}^n h(y_k \mid W, y^{k-1})\nonumber\\
			&= h(y^n) - \sum_{k=1}^n h(y_k \mid W, y^{k-1}, x_k(W, y^{k-1}), x^{k-1}) \label{1}\\
			&= h(y^n) - \sum_{k=1}^n h(y_k \mid W, z^{k-1}, x_k, x^{k-1})\label{2}\\
			&= h(y^n) - \sum_{k=1}^n h(z_k \mid W, z^{k-1}, x_k,  x^{k-1})\label{3}\\
			&= h(y^n) - \sum_{k=1}^n h(z_k \mid W, z^{k-1})\label{4}\\
			&= h(y^n) - h(z^n \mid W)\nonumber \\
			&= h(y^n) - h(z^n) \label{5}
\end{align}
where (\ref{1}) follows from the fact that $x_i = f_i(W, y^{i-1})$ and so $x_i$ is determined by $(W, y^{i-1})$,
(\ref{2}) follows from the equality $y^{k-1} = x^{k-1}+z^{k-1}$, (\ref{3}) follows from $y_k = x_k + z_k$, 
(\ref{4}) follows from the fact that $x^k$ is determined from $W$ and $z^{k-1}$ by the recursion $x_i = f_i(W, x^{i-1}+z^{i-1})$, 
and (\ref{5}) follows from the independence between $z^n$ and $W$. \qed

\subsection{Proof of Theorem \ref{infhorizon_thm}}
\label{thm1}
The achievability is shown in section VII in \cite{cover:1989}. We will here show the converse.
According to Proposition \ref{coverRes}, we have that
$$
\I(W; y^n) = h(y^n) - h(z^n)
$$
Let $\zh_k=\E{z_k |z^{k-1}}$, $\zt_k = z_k - \zh_k$,  $\zb_k=\E{\zh_k |y^{k-1}}$, and $\zc_k = \zh_k - \zb_k$.
Note that $\zt_k$ is independent of $z^{k-1}$ and hence independent of $y^{k-1}$, $x^k$, and $\zb^{k-1}$.  
Let $\xh_k=\E{x_k |y^{k-1}}$, $\xt_k = x_k - \xh_k$, and $$\yt_k = y_k - \xh_k -\zb_k = \xt_k + \zc_k + \zt_k$$
Then,
\begin{eqnarray}
h(y^n) 	&=& \sum_{k=1}^n h(y_k\mid y^{k-1}) \label{hyn}\\
		&=& \sum_{k=1}^n h(\tilde{y}_k\mid y^{k-1}) \label{yhat0}\\
		&\le& \sum_{k=1}^n h(\tilde{y}_k) \label{ytilde0}\\
		&\le& \frac{1}{2}\sum_{k=1}^n \log_2(2\pi e \E{\tilde{y}_k^2}) \label{gaussiany0} \label{logtildey}
\end{eqnarray}
where (\ref{yhat0}) follows from $y_k = \hat{y}_k + \tilde{y}_k$ and $\hat{y}_k$ is a function of $y^{k-1}$, (\ref{ytilde0}) follows from the fact that conditioning only reduces the differential entropy (with equality if $\yt_k$ is independent of $y^{k-1}$) and (\ref{gaussiany0}) follows from the fact that for a fixed covariance, the Gaussian distribution has the maximum entropy.
We will now show that the upper bound of the differential entropy $h(y^n)$ given by (\ref{gaussiany0}) can be achieved by showing
that $$Y_k \triangleq \E{\tilde{y}_k^2}$$ can be achieved by the transmission strategy $x_k = L_k\zc_k + v_k$ for some $L_k\in \R$ and temporally uncorrelated Gaussian variables $\{v_k\}$.
Let $\E{\zt_k^2} = \sigma_k$ and 
\begin{equation}
\label{cov0}
\E{
\begin{pmatrix}
	\zc_k\\
	\xt_k
\end{pmatrix}
\begin{pmatrix}
	\zc_k\\
	\xt_k
\end{pmatrix}^\intercal
} =
\begin{pmatrix}
	\zeta_k		&	\psi_k\\
	\psi_k		&	\xi_k
\end{pmatrix} \succeq 0
\end{equation}
be an achievable covariance matrix by some strategy $x_k$.  
Then,
\begin{align*}
Y_k &=
\E{\tilde{y}_k^2} \\
&= \E{(\xt_k + \zc_k + \zt_k)^2}\\
&= \E{(\xt_k + \zc_k)^2}+\E{\zt_k^2}\\
&= \E{
\begin{pmatrix}
	\zc_k\\
	\xt_k
\end{pmatrix}^\intercal
\begin{pmatrix}
	1	&  1 \\
	1 	&	1
\end{pmatrix}
\begin{pmatrix}
	\zc_k\\
	\xt_k
\end{pmatrix}
} + \sigma_k\\
&=
\E{\tr\lp
\begin{pmatrix}
	1	&  1 \\
	1 	&	1
\end{pmatrix}
\begin{pmatrix}
	\zc_k\\
	\xt_k
\end{pmatrix}
\begin{pmatrix}
	\zc_k\\
	\xt_k
\end{pmatrix}^\intercal\rp
}+ \sigma_k \\
&=
\tr \lp
\begin{pmatrix}
	1 &  1 \\
	1 &  1
\end{pmatrix}
\begin{pmatrix}
	\zeta_k		&	\psi_k\\
	\psi_k		&	\xi_k
\end{pmatrix}\rp + \sigma_k
\end{align*}
The Schur complement in $\xi_k$ of
$$
\begin{pmatrix}
	\zeta_k		&	\psi_k\\
	\psi_k		&	\xi_k
\end{pmatrix}
$$
is given by
$$
\phi_k \triangleq \xi_k - \psi_k^2 \zeta_k^{-1} \ge 0
$$
By taking $x_k = L_k \zc_k + v_k$ with $L_k=\psi_k \zeta_k^{-1}$ and $v$ to be a temporally uncorrelated Gaussian process with $v_k\sim \gv(0, \phi_k)$ independent of $u^k$ and $y^{k-1}$ (and thus independent of $\{\zc_k\}$), we will have $\xh_k = \E{x_k\mid y^{k-1}} = \E{ L_k \zc_k + v_k \mid y^{k-1}} = 0$, since $v_k$ is independent of $y^{k-1}$ and $\zc_k = \zh_k - \E{\zh_k |y^{k-1}}$ which is also independent of $y^{k-1}$. 

Then, $\xt_k = x_k - \xh_k = x_k$, and we will 
get a sequence of pairs $(\zc_k, x_k)$ that renders the covariance matrix (\ref{cov0}) since
\begin{equation}
\begin{aligned}
&\E{
\begin{pmatrix}
	\zc_k\\
	\xt_k
\end{pmatrix}
\begin{pmatrix}
	\zc_k\\
	\xt_k
\end{pmatrix}^\intercal
} \\
&=\E{
\begin{pmatrix}
	\zc_k\\
	\psi_k \zeta_k^{-1} \zc_k + v_k
\end{pmatrix}
\begin{pmatrix}
	\zc_k\\
	\psi_k \zeta_k^{-1} \zc_k + v_k
\end{pmatrix}^\intercal
} \\
&=\begin{pmatrix}
	\zeta_k		&	\psi_k\\
	\psi_k		&	\psi_k^2 \zeta_k^{-1} + \phi_k
\end{pmatrix}\\
&=\begin{pmatrix}
	\zeta_k		&	\psi_k\\
	\psi_k		&	\xi_k
\end{pmatrix} 
\end{aligned}
\end{equation} 
Also, since this strategy is linear and $v$ is a Gaussian process, $y$ is also a Gaussian process and hence, $\yt$ is Gaussian such that $\yt_k$ is independent of $y^{k-1}$, and the entropy upper bound (\ref{gaussiany0}) is achieved.

Hence, an optimal strategy has the form
$$
x_k = \sB_k(\z)z_{k}+v_k
$$ 
with
$$\sB_k(\z)=\sum_{l=1}^{\infty} b_l(k)\z^{-l}$$ 
and $v$ is a white Gaussian process. 

Now for the stationary noise process $z$, $\sigma_k$ will be constant, $\sigma_k = \sigma$. The upperbound (\ref{logtildey}) is maximized when
 $\E{\tilde{y}_1^2} = \cdots = \E{\tilde{y}_n^2}$ since $\log_2(\cdot)$ is a concave function which implies that
\begin{equation}
\label{am-gm}
\begin{aligned}
\frac{1}{n}h(y^n) &= \frac{1}{2n}\sum_{k=1}^n \log_2(2\pi e \E{\tilde{y}_k^2})\\
			&\le \frac{1}{2}\log_2\left(2\pi e \frac{1}{n}\sum_{k=1}^n \E{\tilde{y}_k^2}\right)
\end{aligned}
\end{equation}
with equality if and only if $\E{\tilde{y}_1^2} = \cdots = \E{\tilde{y}_n^2}$. Now $\E{\tilde{y}_k^2} = (L_k+1)^2\zeta_k + \phi_k + \sigma$, and by taking $L_k = L$ and $\phi_k = V$, we get a stationary process $\zc$ and so $\zeta_k = \zeta$.
Thus, we get  $\E{\tilde{y}_1^2} = \cdots = \E{\tilde{y}_n^2}$. 
This implies that we can take $\sB_k$ to be constant, $\sB_k = \sB$, for all $k$.
The spectral density function of the 
stationary Gaussian process $y$ is given by
$$
\sS_y(\z) = V + |\sB(\z) + 1|^2 \sS_z(\z)
$$
where $V>0$. Note that the strict inequality is necessary since $V=0$ implies that $v$ vanishes (almost everywhere), which implies in turn that
the mutual information between $v$ and $y$ vanishes and it would render would render a zero achievable
rate. Now the average power of the process $x$ is given by
\[
V+ \int_{-\pi}^{\pi}  |\sB(e^{i \theta})|^2\sS_z(e^{i \theta})
		\frac{d\theta}{2\pi} \le P
\]
Furthermore, according to Proposition \ref{entropy}, the differential entropies of $y$ and $z$ are given by
$$
h(y) = \int_{-\pi}^{\pi} \frac{1}{2} \log \left(
		2\pi e\left(V+ |\sB(e^{i \theta})+1|^2 \sS_z(e^{i \theta})\right) \right) 
		\frac{d\theta}{2\pi}
$$
and
$$
h(z) =  \int_{-\pi}^{\pi} \frac{1}{2} \log \left(
		2\pi e \sS_z(e^{i \theta})\right)
		\frac{d\theta}{2\pi}
$$
Now the average feedback channel capacity is given by
\begin{align*}
C 	
			&= \lim_{n\rightarrow \infty} \frac{1}{n}h(y^n) - \frac{1}{n}h(z^n)\\ 
			&=  h(y) - h(z)
\end{align*}

Hence, the average feedback channel capacity $C$ is the optimal value of the optimization problem
\begin{equation*}
	\begin{aligned}
		\sup_{\substack{V>0\\ \sB(\z)= \sum_{l=1}^\infty b_l \z^{-l}}} & \int_{-\pi}^{\pi} \frac{1}{2} \log 
		\frac{V+ |\sB(e^{i \theta})+1|^2 \sS_z(e^{i \theta}) }{
		\sS_z(e^{i \theta})
		}
		\frac{d\theta}{2\pi}\\
		\textup{s.t.  }~~~~ & V+ \int_{-\pi}^{\pi}  |\sB(e^{i \theta})|^2\sS_z(e^{i \theta})
		\frac{d\theta}{2\pi} \le P
	\end{aligned} 
\end{equation*}
\qed

\subsection{Proof of Theorem \ref{optimalTx}}
\label{thm2}
According to Proposition \ref{coverRes}, we have that
$$
\I(W; y^n) = h(y^n) - h(z^n)
$$
Let $\zh_k=\E{z_k |z^{k-1}}$ and $\zt_k = z_k - \zh_k$.
We have that 
$$
h(z^n) = \sum_{k=1}^n h(z_k\mid z^{k-1}) =  \sum_{k=1}^n h(\tilde{z}_k)
$$
Since $s_1 =  0$, we get $\tilde{z}_1 = z_1 = u_1$. We also have that 
$\E{z_2\mid z_1} = Hs_2$, and so $\tilde{z}_2 = u_2$. Inductively, we get that
$\tilde{z}_k = u_k$ for $k=1, ..., n$. Thus,
$$
h(z^n) = \sum_{k=1}^n h(\tilde{z}_k)=\sum_{k=1}^n h(u_k) = \frac{1}{2}\sum_{k=1}^n \log_2(2\pi e)
$$

Let $\sh_k=\E{s_k |y^{k-1}}$, $\st_k = s_k - \sh_k$, and $$\yt_k = y_k -H\sh_k= x_k + H\st_k + u_k.$$
Note that similar to the proof of Theorem \ref{infhorizon_thm}
we have that
\begin{eqnarray}
h(y^n) 	
		&\le& \frac{1}{2}\sum_{k=1}^n \log_2(2\pi e \E{\tilde{y}_k^2}) \label{gaussiany}
\end{eqnarray}
We will now show that the upper bound of the differential entropy $h(y^n)$ given by (\ref{gaussiany}) can be achieved by showing
that $$Y_k \triangleq \E{\tilde{y}_k^2}$$ can be achieved by the transmission strategy $x_k = X_k\st_k + v_k$ for some $X_k\in \R$ and temporally uncorrelated Gaussian variables $\{v_k\}$.

First note that $\E{\tilde{y}_k^2} = \E{(x_k + H\st_k + u_k)^2}=\E{(x_k + H\st_k)^2} + \E{u_k^2}$ since $u_k$ is 
independent of $x_k$ and $\st_k$.
Now let
\begin{equation}
\label{cov}
\E{
\begin{pmatrix}
	\st_k\\
	x_k
\end{pmatrix}
\begin{pmatrix}
	\st_k\\
	x_k
\end{pmatrix}^\intercal
} =
\begin{pmatrix}
	\Sigma_k		&	\Psi_k\\
	\Psi^\T_k		&	\Xi_k
\end{pmatrix}
\end{equation}
be an achievable covariance matrix by some strategy $x_k$.  
Then,
\begin{align*}
Y_k &=
\E{\tilde{y}_k^2} \\
&= \E{(x_k + H\st_k)^2} + \E{u_k^2}\\
&= \E{
\begin{pmatrix}
	\st_k\\
	x_k
\end{pmatrix}^\intercal
\begin{pmatrix}
	H^\T H	& H^\T \\
	H 		&	1
\end{pmatrix}
\begin{pmatrix}
	\st_k\\
	x_k
\end{pmatrix}
} + 1\\
&=
\E{\tr\lp
\begin{pmatrix}
	H^\T H	& H^\T \\
	H 		&	1
\end{pmatrix}
\begin{pmatrix}
	\st_k\\
	x_k
\end{pmatrix}
\begin{pmatrix}
	\st_k\\
	x_k
\end{pmatrix}^\intercal\rp
}+ 1\\
&=
\tr \lp
\begin{pmatrix}
	H^\T H	& H^\T \\
	H 		&	1
\end{pmatrix}
\begin{pmatrix}
	\Sigma_k		&	\Psi_k\\
	\Psi^\T_k		&	\Xi_k
\end{pmatrix}\rp + 1
\end{align*}
The Schur complement in $\Xi_k$ of
$$
\begin{pmatrix}
	\Sigma_k		&	\Psi_k\\
	\Psi^\T_k		&	\Xi_k
\end{pmatrix}\succeq 0
$$
is given by
$$
\Phi_k \triangleq \Xi_k - \Psi_k^\T \Sigma_k^{-1} \Psi_k \succeq 0
$$
By taking $x_k = X_k \st_k + v_k$ with $X_k=\Psi_k^\T \Sigma_k^{-1}$ and $v_k\sim \gv(0, \Phi_k)$ independent of $\st_k$, $u_k$, and $v_l$ for $l\ne k$, we will get a sequence of pairs $(\st_k, x_k)$ that renders the covariance matrix (\ref{cov}). Also, since this strategy is linear and $v$ is a Gaussian process, $\yt$ is Gaussian such that $\yt_k$ is independent of $y^{k-1}$, and the differential entropy upper bound (\ref{gaussiany}) is achieved. The mutual information becomes
\begin{align*}
\I(W; y^n) 	&= h(y^n) - h(z^n)\\ 
			&= \frac{1}{2}\sum_{k=1}^n \log_2(2\pi e \E{\tilde{y}_k^2}) - \frac{1}{2}\sum_{k=1}^n \log_2(2\pi e)\\
			&= \frac{1}{2}\sum_{k=1}^n \log_2(2\pi e Y_k) - \frac{1}{2}\sum_{k=1}^n \log_2(2\pi e)\\
			&=  \frac{1}{2}\sum_{k=1}^n \log_2(Y_k)
\end{align*}
and
$$
C_n = \sup_{f} ~ \frac{1}{n} \I(W; y^n ) = \max_{\substack{X_1, ..., X_n\\ V_1, ..., V_n\ge 0}} \frac{1}{2n}\sum_{k=1}^n \log_2(Y_k)
$$
Now for $x_k = X_k\st_k + v_k$, we have that
$$\yt_k = y_k -H\sh_k = (X_k+H)\st_k + v_k + u_k$$ 
Let 
\begin{align*}
\Gamma_k 
		&\triangleq \E{(F\st_k + Gu_k)\yt_k}(\E{\yt_k^2})^{-1}\\ 
		&= (F\Sigma_k (X_k +H)^\T + G)Y_k^{-1}
\end{align*}
Then,
$$
\E{F\st_k + Gu_k | \yt_k} = \Gamma_k \yt_k
$$
The dynamics of $\sh_k$ and $\st_k$ are given by

\begin{eqnarray}
	\sh_{k+1} &=& \E{Fs_k + Gu_k | y^{k}}\\
				&=& \E{Fs_k + Gu_k | y_k, y^{k-1}}\\
				&=& \E{F(\sh_k+\st_k) + Gu_k | \yt_k, y^{k-1}}\\
				&=&  F\sh_k + \E{F\st_k + Gu_k | \yt_k} \label{forth}\\
				&=& F\sh_k + \Gamma_k \yt_k
\end{eqnarray}
where (\ref{forth}) follows from the orthogonality between $(u_k, \st_k, \yt_k)$ and $y^{k-1}$, which implies
independence since all variables are jointly Gaussian.
Hence, the error dynamics become
\[
\begin{aligned}
	\st_{k+1} &= F\st_{k} - \Gamma_k \yt_k + G u_k  \\
	\yt_{k}	&= (X_k+H)\st_k + v_k + u_k
\end{aligned}
\]
which implies that
\begin{equation*}
Y_k = (X_k+H)\Sigma_k (X_k+H)^\T +  V_k + 1
\end{equation*}
Finally, we have that $\Sigma_1 = \E{\st_1 \st_1^\T} = \E{s_1 s_1^\T} = 0$ and the recursion of $\Sigma_k$
is given by
\begin{equation*}
\begin{aligned}
\Sigma_{k+1} 	&= (F-\Gamma_k(X_k+H))\Sigma_k (F-\Gamma_k(X_k+H))^\T \\
			& ~~~ + (G-\Gamma_k)(G-\Gamma_k)^\T + \Gamma_k V_k \Gamma_k^\T\\ 
\end{aligned}
\end{equation*}
Now let $S_k = \E{s_k s_k^\intercal}$ and $\hat{S}_k=\E{\hat{s}_k\hat{s}_k^\intercal}$. Then, $S_k = \hat{S}_k + \Sigma_k$ (since $\hat{s}_k$ and $\tilde{s}_k$
are independent).
Also, since $S_{k+1}=FS_kF^\intercal +GG^\T$ and $\hat{S}_{k+1}= F\hat{S}_kF^\T + \Gamma_k Y_k \Gamma_k^\T$, we obtain
\begin{align*}
\Sigma_{k+1} &= S_{k+1} - \hat{S}_{k+1}\\
					&= FS_kF^\intercal +GG^\T - F\hat{S}_kF^\T - \Gamma_k Y_k \Gamma_k^\T\\
					&= F\Sigma_kF^\intercal +GG^\T- \Gamma_k Y_k \Gamma_k^\T
\end{align*} \qed

\subsection{Proof of Theorem \ref{kim}}
\label{thm3}
We know that the recursion given by (\ref{Yrec})--(\ref{riccrec}) for the optimal choice of $(X_k, V_k)$ converge since the suboptimal choice $X_k=0$ implies a converging recursion for any bounded sequence $\{V_k\}$ by the stability of the matrix $F$ (\cite{zhou:1996}).

Now at stationarity, the variance of $\yt_k$ is given by
\begin{equation}
\label{Y}
Y_k  = Y  = (X+H)\Sigma (X+H)^\T +  V + 1
\end{equation}
where for all $k$,
$$X_k=X,$$
$$V_k=V,$$
$$\Gamma_k = \Gamma,$$
$$\Sigma_k = \Sigma,$$
$\Sigma$ is the unique solution to the Riccati equation
\begin{equation}
\label{ricc}
\begin{aligned}
\Sigma 	&= (F-\Gamma(X+H))\Sigma (F-\Gamma(X+H))^\T \\
			& ~~~ + (G-\Gamma)(G-\Gamma)^\T + \Gamma V \Gamma^\T\\ 
			&= F\Sigma F^\T + GG^\T - \Gamma Y \Gamma^\T 
\end{aligned}
\end{equation}
and
$$
\Gamma = (F\Sigma (X+H)^\T + G)Y^{-1}.
$$
Recall that the power constraint at stationarity is given by
$$P\geq \E{x_k^2} = X\Sigma X^\T + V. $$
The average feedback channel capacity is then given by
$$C=\lim_{n\rightarrow \infty} C_n = \frac{1}{2}\log_2(Y)$$
and we arrive at the following optimization problem for finding the maximum capacity at stationarity:
\begin{equation}
\label{nonnegV0}
\begin{aligned}
 	\sup_{V>0} \sup_{\substack{X \\ \Sigma \succeq 0}} &~~ \frac{1}{2}\log_2(Y)\\
	\text{s.t. }   ~~ P &\ge X\Sigma X^\T + V \\
				 ~~ \Sigma &= F\Sigma F^\T + GG^\T - \Gamma Y \Gamma^\T \\
				 ~~ \Gamma &= (F\Sigma (X+H)^\T + G)Y^{-1} \\
				 ~~ Y &= (X+H)\Sigma (X+H)^\T +  V + 1
\end{aligned}
\end{equation}
Since $Y$ is a scalar, we may maximize $Y$ instead of its
logarithm and we get the equivalent optimization
problem (\ref{nonnegV}).

\subsection{Proof of Lemma \ref{epslemma}}
\label{lem1}
Consider a stationary finite order Gaussian noise process $z$ given by
\begin{equation}
\label{linsyseps}
\begin{aligned}
s_{k+1} 			&= Fs_k + Gu_k\\
z_k					&= Hs_k + u_k\\
u_k				&\sim \gv(0, 1)\\
\E{u_k u_l} 	&= 0 ~~ \forall k, l ~|~ k\neq l	\\
\end{aligned}
\end{equation}
In the proof of Theorem \ref{optimalTx} we have derived the equation
\begin{align*}
\Sigma_{k+1} &= F\Sigma_kF^\intercal +GG^\T- \Gamma_k Y_k \Gamma_k^\T
\end{align*}
which describes the estimation error dynamics of the Kalman filter. It's well known that
the error covariance is increasing, that is 
$\Sigma_{k+1}\succeq \Sigma_{k}$. 
Thus,
\begin{align*}
 \Sigma_{k} \preceq \Sigma_{k+1} =  F\Sigma_kF^\intercal +GG^\T- \Gamma_k Y_k \Gamma_k^\T
\end{align*}
and at stationarity, we have that
\[
\Sigma \preceq  F\Sigma F^\T + GG^\T - \Gamma Y \Gamma^\T
\]
Following the same lines as Theorem \ref{kim}, the average feedback capacity with the noise process given by
(\ref{linsyseps}) is $$C = \frac{1}{2}\log_2(Y)$$
where $Y$ is the solution to the optimization problem
\begin{equation*}
\begin{aligned}
 	\sup_{V>0} \sup_{\substack{X \\ \Sigma \succeq 0}} &~~ Y\\
	\text{s.t. }   ~~ P &\ge X\Sigma X^\T + V \\
				 ~~ \Sigma &\preceq  F\Sigma F^\T + GG^\T + \epsilon I - \Gamma Y \Gamma^\T \\
				 ~~ \Gamma &= (F\Sigma (X+H)^\T + G)Y^{-1} \\
				 ~~ Y &= (X+H)\Sigma (X+H)^\T +  V + 1
\end{aligned}
\end{equation*}
 and the proof is complete. \qed

\subsection{Proof of Lemma \ref{controllablepair}}
\label{lem2}
Note first that from Proposition \ref{contcond}, the pair
$$
(F-\Gamma(X+H), (G-\Gamma)(G-\Gamma)^\T + \Gamma V \Gamma^\T)
$$
is controllable if there does not exist a complex number $\lambda \in \C$ and a vector $x$ such that
$$
x^\T \lp (G-\Gamma)(G-\Gamma)^\T + \Gamma V \Gamma^\T \rp = 0
$$
and
$$
 x^\T (F-\Gamma(X+H)) = \lambda x^\T
$$
Now suppose that $x$ is such that 
$$x^\T \lp (G-\Gamma)(G-\Gamma)^\T + \Gamma V \Gamma^\T \rp= 0$$
Then,
$$
x^\T \lp (G-\Gamma)(G-\Gamma)^\T + \Gamma V \Gamma^\T \rp x= 0
$$
Since $(G-\Gamma)(G-\Gamma)^\T \succeq 0$ and $\Gamma V \Gamma^\T\succeq 0$, we must have that
$x^\T \Gamma V \Gamma^\T x= V |\Gamma^\T x|^2 = 0$ and  $x^\T (G-\Gamma)(G-\Gamma)^\T x= 0$. Since $V>0$,  
we must have $\Gamma^\T x = 0$. Thus, $x^\T \Gamma = (\Gamma^\T x)^\T = 0$ . Similarly, the equality $x^\T (G-\Gamma)(G-\Gamma)^\T x= 0$ implies that 
$x^\T(G-\Gamma) = 0$ and since $x^\T \Gamma = 0$, we get $x^\T G = 0$. Now $(F, G)$ is controllable, and Proposition \ref{contcond} implies that there does not exist a number $\lambda \in \C$ such that  $x^\T G = 0$ and $\lambda x^\T = x^\T F = x^\T (F-\Gamma(X+H)) $. Thus, we cannot have a vector $x$ and a number $\lambda \in \C$ such that
 $$
x^\T \lp (G-\Gamma)(G-\Gamma)^\T + \Gamma V \Gamma^\T \rp = 0
$$
and
$$
 x^\T (F-\Gamma(X+H)) = \lambda x^\T
$$
Hence, using  Proposition \ref{contcond} again, we conclude that
$$
(F-\Gamma(X+H), (G-\Gamma)(G-\Gamma)^\T + \Gamma V \Gamma^\T)
$$
is controllable. \qed

\subsection{Proof of Lemma \ref{stableF}}
\label{lem3}

Since the pair  $(F, G)$ is controllable and $V>0$, 
Lemma \ref{controllablepair} implies that
$$
(F-\Gamma(X+H), (G-\Gamma)(G-\Gamma)^\T + \Gamma V \Gamma^\T)
$$
is controllable. Suppose that there exists a vector $x$ such that $x^\T(F-\Gamma(X+H)) =x^\T \lambda$
for $|\lambda|\ge 1$. Then,
\begin{equation*}
\begin{aligned}
x^\T\Sigma x	&= x^\T(F-\Gamma(X+H))\Sigma (F-\Gamma(X+H))^\T x \\
				& ~~~ + x^\T(G-\Gamma)(G-\Gamma)^\T x + x^\T \Gamma V \Gamma^\T x\\ 
				&= |\lambda|^2 x^\T \Sigma x + x^\T(G-\Gamma)(G-\Gamma)^\T x + x^\T \Gamma V \Gamma^\T x\\
				&\ge x^\T \Sigma x + x^\T(G-\Gamma)(G-\Gamma)^\T x + x^\T \Gamma V \Gamma^\T x\\
\end{aligned}
\end{equation*}
which implies that $x^\T((G-\Gamma)(G-\Gamma)^\T + \Gamma V \Gamma^\T)=0$. But then, we obtain that
$$
x^\T (F-\Gamma(X+H) -\lambda I ~~~ (G-\Gamma)(G-\Gamma)^\T + \Gamma V \Gamma^\T) = 0
$$
which contradicts the fact that the pair
$$
(F-\Gamma(X+H), (G-\Gamma)(G-\Gamma)^\T + \Gamma V \Gamma^\T)
$$
is controllable. Thus, we conclude that $|\lambda|<1$ and the closed loop matrix $F-\Gamma(X+H)$ must be stable. \qed

\subsection{Proof of Lemma \ref{invertibleS}}
\label{lem4}
Since the pair  $(F, G)$ is controllable and $V>0$, 
Lemma \ref{controllablepair} implies that
$$
(F-\Gamma(X+H), (G-\Gamma)(G-\Gamma)^\T + \Gamma V \Gamma^\T)
$$
is controllable and
Lemma \ref{stableF} implies that
$F-\Gamma(X+H)$
is stable. Taking $Q=(G-\Gamma)(G-\Gamma)^\T + \Gamma V \Gamma^\T\succeq 0$, we conclude from
Proposition \ref{riccinv} that $\Sigma$ is invertible and thus, $\Sigma\succ 0$. \qed

\subsection{Proof of Theorem \ref{main}}
\label{thm4}

We can make the variable substitution $K = X\Sigma$ and maximize with respect to $V$, $K$ and $\Sigma \succ 0$. This gives the optimization problem 
\begin{equation}
\label{lmi1}
	\begin{aligned}
	\sup_{V>0} \sup_{\substack{K \\ \Sigma \succ 0}} &~~ Y\\
	\text{s.t. }	 ~~ P &\ge K\Sigma^{-1} K^\T + V \\
				 ~~ \Sigma &\preceq F\Sigma F^\T + GG^\T - \Gamma Y \Gamma^\T \\
				 ~~ \Gamma &= (FK^\T + F\Sigma H^\T + G)Y^{-1} \\
				  ~~ Y &= K\Sigma^{-1}K^\T + KH^\T + HK^\T + H\Sigma H^\T\\ 
				 ~~ &~~~ +V+1  
	\end{aligned}
\end{equation}
First we show the intuitive result that for any $K$ and $\Sigma$, taking $V$ such that equality is achieved for the power constraint, $P =  K\Sigma^{-1} K^\T + V$, is optimal. Suppose that $V$ is such that $P >  K\Sigma^{-1} K^\T + V$. Introduce $V' = P - K\Sigma^{-1} K^\T > V$. 
Then, we have that
\[
\begin{aligned}
Y 	&\le K \Sigma^{-1}K^\T + K H^\T + HK^\T + H\Sigma H^\T + V' + 1 \\
	&= K H^\T + HK^\T + H\Sigma H^\T + P + 1
\end{aligned}
\]
Furthermore, with $\Gamma' \triangleq (FK^\T + F\Sigma H^\T + G)(Y')^{-1}$, 
$
Y' = K\Sigma^{-1}K^\T + KH^\T + HK^\T + H\Sigma H^\T + V' + 1
$, 
the constraint 
\[
\begin{aligned}
	\Sigma 	&\preceq F\Sigma F^\T + GG^\T - \Gamma Y \Gamma^\T \\
				&\preceq F\Sigma F^\T + GG^\T - \Gamma' Y' (\Gamma')^\T \\
\end{aligned}
\]
is satisfied. Thus, increasing $V$ to $V'$ increases the value of $Y$ to $Y'$.  Hence, (\ref{lmi1}) is equivalent to
\begin{equation}
\label{lmi12}
	\begin{aligned}
	\sup_{V>0} \sup_{\substack{K \\ \Sigma \succ 0}} &~~ Y\\
	\text{s.t. }	 ~~ P &= K\Sigma^{-1} K^\T + V \\
				 ~~ \Sigma &\preceq F\Sigma F^\T + GG^\T - \Gamma Y \Gamma^\T \\
				 ~~ \Gamma &= (FK^\T + F\Sigma H^\T + G)Y^{-1} \\
				  ~~ Y &= KH^\T + HK^\T + H\Sigma H^\T +P+1  
	\end{aligned}
\end{equation}

Now we have that
$$
P \ge K\Sigma^{-1} K^\T + V > K\Sigma^{-1} K^\T
$$

Thus, any parameters satisfying the constraints in (\ref{lmi12}) also satisfy the constraints in 
\begin{equation}
\label{lmi2}
\begin{aligned}
\sup_{\substack{K \\ \Sigma \succ 0}} &~~ Y\\
	\text{s.t. }	& ~~ P > K\Sigma^{-1} K^\T \\
				& ~~ \Pi = FK^\T + F\Sigma H^\T + G \\
				& ~~ Y = KH^\T + HK^\T + H\Sigma H^\T +P +1 \\
				& ~~ \Sigma \preceq F\Sigma F^\T + GG^\T - \Pi Y^{-1} \Pi^\T \\
\end{aligned}
\end{equation}

Similarly, if $K, \Sigma$ satisfy the constraints in (\ref{lmi2}), then they also satisfy the constraints in (\ref{lmi12}) with 
$V = P - K\Sigma^{-1} K^\T$.

The power inequality 
$$P > K\Sigma^{-1}K^\T$$
is equivalent to the linear matrix inequality (LMI)
$$
\begin{pmatrix}
	P			&  K\\
	K^\T	& \Sigma
\end{pmatrix}\succ 0
$$
The error covariance inequality 
\[
\begin{aligned}
				& ~~ \Sigma \preceq F\Sigma F^\T + GG^\T - \Pi Y^{-1} \Pi^\T \\
\end{aligned}
\]
can be recast as the LMI
\[
\begin{aligned}
0 	&\preceq
	\begin{pmatrix}
		F\Sigma F^\T - \Sigma + GG^\T 		& \Pi\\
		\Pi^\T	& Y
	\end{pmatrix} \\
	&=
	\begin{pmatrix}
		F\Sigma F^\T - \Sigma + GG^\T 		& FK^\T + F\Sigma H^\T + G\\
		(FK^\T + F\Sigma H^\T + G)^\T	& Y
	\end{pmatrix} 
\end{aligned}
\]

Summing up, the maximum differential entropy of the channel output is given by the value of the following optimization (maximization) problem
\begin{equation}
\label{maxminopt2}
\begin{aligned}
 	\sup_{\substack{K \\ \Sigma \succ 0}}&~~   Y \\
	\textup{s.t.  } \\   0 \prec & \begin{pmatrix}
										P		&  K\\
										K^\T	& \Sigma
									\end{pmatrix}\\
					~~ Y = & KH^\T + HK^\T + H\Sigma H^\T +P +1 \\
					~~	 0 \preceq & 
					\begin{pmatrix}
						F\Sigma F^\T - \Sigma + GG^\T 		& FK^\T + F\Sigma H^\T + G\\
						(FK^\T + F\Sigma H^\T + G)^\T	& Y
					\end{pmatrix}
\end{aligned}
\end{equation}
\qed


\subsection{Matlab Code}
\label{code}
\begin{verbatim}
alpha = 0.7;
beta = -0.25;
sigma = sign(beta-alpha);
P = 1;
d = 1;
F = -beta;
G = 1;
H = alpha-beta;

a4 = alpha^2+beta^2*P;
a3 = 2*sigma*(alpha+beta*P);
a2 = P+1-alpha^2;
a1 = -2*sigma*alpha;
a0 = -1;
r = roots([a4 a3 a2 a1 a0]);

cvx_begin sdp
    variable S(d,d) symmetric
    variable K(1,d)
    variable Y
    S > 0
    [P K; K' S] > 0
    [F*S*F' - S + G*G'  F*K' + F*S*H' + G; 
        K*F' + H*S*F' + G'  Y] > 0
    Y == K*H' + H*K' + H*S*H' + P + 1
    maximize Y;
cvx_end

C_sdp = 0.5*log2(Y)
C_poly = -log2(r(4))
\end{verbatim}

\bibliography{../ref/mybib}

\end{document}